\documentclass{article}

\usepackage{microtype}
\usepackage{graphicx}
\usepackage{subfigure}
\usepackage{booktabs} 
\usepackage{amssymb} 

\usepackage{hyperref}

\usepackage[accepted]{icml2020}

\usepackage[normalem]{ulem} 

\icmltitlerunning{Boomerang Sampler}

\usepackage{amsfonts}
\usepackage{amsmath}
\usepackage{amsthm}
\usepackage{bm}
\usepackage{enumerate} 

\newcommand{\ind}{\mathbf{1}}
\newcommand{\dd}{ \mathrm{d}}
\newcommand{\dif}{\dd}

\newcommand{\E}{\mathbb E}
\renewcommand{\P}{\mathbb P}
\newcommand{\R}{\mathbb R}
\newcommand{\SSigma}{\bm \Sigma}

\newcommand{\x}{\bm x}
\newcommand{\y}{\bm y}
\newcommand{\z}{\bm z}
\renewcommand{\v}{\bm v}
\newcommand{\w}{\bm w}

\newcommand{\refr}{\textrm{refr}}
\newcommand{\0}{\bm 0}
\DeclareMathOperator*{\argmin}{arg\,min}

\theoremstyle{plain}
\newtheorem{theorem}{Theorem}[section]

\newtheorem{lemma}[theorem]{Lemma}
\newtheorem{proposition}[theorem]{Proposition}

\theoremstyle{remark}
\newtheorem{remark}[theorem]{Remark}

\begin{document}
	
	\twocolumn[
	\icmltitle{The Boomerang Sampler}
	
	
	
	
    \begin{icmlauthorlist}
		\icmlauthor{Joris Bierkens}{delft}
		\icmlauthor{Sebastiano Grazzi}{delft}
		\icmlauthor{Kengo Kamatani}{osaka}
		\icmlauthor{Gareth Roberts}{warwick}
	\end{icmlauthorlist}
	\icmlaffiliation{delft}{Technische Universiteit Delft, Netherlands}
	\icmlaffiliation{osaka}{Graduate School of Engineering Science, Osaka University, Japan}
	\icmlaffiliation{warwick}{Department of Statistics, University of Warwick, United Kingdom}
	
	\icmlcorrespondingauthor{Joris Bierkens}{joris.bierkens@tudelft.nl}

	\icmlkeywords{MCMC, piecewise deterministic Monte Carlo, Bayesian inverse problems, diffusion bridge simulation}

	\vskip 0.3in
	]



\printAffiliationsAndNotice{}  

\begin{abstract}
This paper introduces the Boomerang Sampler as a novel class of continuous-time non-reversible Markov chain Monte Carlo algorithms. The methodology begins by representing the target density as a density, $e^{-U}$, with respect to a prescribed (usually) Gaussian measure and constructs a continuous trajectory consisting of a piecewise elliptical path. The method moves from one elliptical orbit to another according to a rate function which can be written in terms of $U$. We demonstrate that the method is easy to implement and demonstrate empirically that it can out-perform existing benchmark piecewise deterministic Markov processes such as the bouncy particle sampler and the Zig-Zag. In the Bayesian statistics context, these competitor algorithms are of substantial interest in the large data context due to the fact that they can adopt data subsampling techniques which are exact (ie induce no error in the stationary distribution). We demonstrate theoretically and empirically that we can also construct a control-variate subsampling boomerang sampler which is also exact, and which possesses remarkable scaling properties in the large data limit.
We furthermore illustrate a factorised version on the simulation of diffusion bridges.
\end{abstract}

\section{Introduction}

Markov chain Monte Carlo remains the gold standard for asymptotically exact (ie bias-free) Bayesian inference for complex problems in Statistics and Machine Learning; see for example \cite{brooks2011handbook}. Yet a major impediment to its routine implementation for large data sets is the need to evaluate the target density (and possibly other related functionals) at each algorithm iteration.

Partly motivated by this, in recent years there has been a surge in the development of innovative piecewise deterministic Monte Carlo methods (PDMC, most notably the Bouncy Particle Sampler (BPS) \cite{BouchardCoteVollmerDoucet2017} and the Zig-Zag Sampler (ZZ) \cite{BierkensFearnheadRoberts2016}), as a competitor for classical MCMC algorithms such as Metropolis-Hastings and Gibbs sampling. We refer to \cite{Fearnhead2016a} for an accessible introduction to the PDMC setting. The primary benefits of these methods are the possibility of \emph{exact subsampling} and \emph{non-reversibility}. Exact subsampling refers to the possibility of using only a subset of the full data set (or even just a single observation) at each iteration of the algorithm, without introducing bias in the output of the algorithm \cite{Fearnhead2016a}. Non-reversibility is a property of MCMC algorithms related to a notion of direction of the algorithm, reducing the number of backtracking steps, thus reducing the diffusivity of the algorithm and reducing the asymptotic variance; as analyzed e.g. in \cite{MR1789978,andrieu2019peskun}.

The current key proponents BPS and ZZ of the PDMC paradigm share the following description of their dynamics. The process moves continuously in time according to a constant velocity over random time intervals, which are separated by `switching events'. These switching events occur at stochastic times at which the velocity, or a component of it, is either reflected, or randomly refreshed. The direction of a reflection, and the random time at which it occurs, is influenced by the target probability distribution.

In this paper we explore the effect of modifying the property of constant velocity. By doing so we introduce the Boomerang Sampler which has dynamics of the simple form $\frac{\dd \x}{\dd t} = \v$, $\frac{\dd \v}{\dd t} = -\x$. Similar ideas were introduced in
\cite{vanetti2017piecewisedeterministic} and termed Hamiltonian-BPS,
a method which can be seen as a special case of our approach.
We generalise the Hamiltonian-BPS algorithm in three important ways.
\begin{enumerate}
    \item 
We relax a condition which restricts the covariance function of the auxiliary velocity process to be isotropic. This generalisation is crucial to ensure good convergence properties of the algorithm.
\item
Furthermore we extend the Boomerang Sampler to allow for exact subsampling (as introduced above), thus permitting its application efficiently for large data sets.
\item
We also introduce a factorised extension of the sampler which has important computational advantages  in the common situation where the statistical model exhibits suitable conditional dependence structure.
\end{enumerate}
Our method also has echoes of the elliptical slice sampler \cite{pmlr-v9-murray10a} which has been a successful discrete-time MCMC method especially within machine learning applications.
Both methods are strongly motivated by Hamiltonian dynamics although there are also major differences in the two approaches.
Finally we mention some other PDMP methods with non-linear dynamics such as Randomized HMC \cite{bou-rabee2017,Deligiannidis2018}, and others \cite{markovic2018bouncy,Terenin2018}.

We shall study the Boomerang Sampler and two subsampling alternatives theoretically by analysing the interaction of Bayesian posterior contraction, data size ($n$) and subsampling schemes in the regular (smooth density) case. We shall show that no matter the rate of posterior contraction, a suitably constructed subsampled Boomerang sampler achieves an ${\cal O} (n)$ advantage over non-subsampled algorithms. 

At the same time, we show that for the (non-subsampled) Boomerang Sampler, the number of switching events, and thus the computational cost, can be reduced by factor $\mathcal{ O}(1/d)$ (where $d$ is the number of dimensions) relative to other piecewise deterministic methods, thanks to the deterministic Hamiltonian dynamics of the Boomerang Sampler.

We illustrate these analyses with empirical investigations in which we compare the properties of Boomerang samplers against other PDMC benchmarks demonstrating the superiority of subsampled Boomerang for sufficiently large data size for any fixed dimension in the setting of logistic regression. We shall also give an empirical study to compare the Boomerang Sampler with its competitors as dimension increases.
Finally, as a potentially very useful application we describe the simulation of diffusion bridges using the Factorised Boomerang Sampler, demonstrating substantial computational advantages compared to its natural alternatives.

\subsection*{Notation}
For $\bm a \in \R^d$ and $\bm \Sigma$ a positive definite matrix in $\R^{d \times d}$ we write $\mathcal N(\bm a,\bm \Sigma)$ for the Gaussian distribution in $\R^d$ with mean $\bm a$ and covariance matrix $\bm \Sigma$. Let $\langle \cdot, \cdot \rangle$ denote the Euclidean inner product in $\R^{d}$. We write $(a)_+ := \max(a,0)$ for the positive part of $a \in \R$, and we write $\langle \cdot, \cdot \rangle_+ := (\langle \cdot, \cdot \rangle)_+$ for the positive part of the inner product.

\section{The Boomerang Sampler}

The Boomerang Sampler is a continuous time, piecewise deterministic Markov process (PDMP) with state space $S = \R^d \times \R^d$.  The two copies of $\R^d$ will be referred to as the \emph{position space} and the \emph{velocity space}, respectively. Our primary interest is in sampling the position coordinate, for which the auxiliary velocity coordinate is a useful tool for us.

Let $\mu_0$ denote a Gaussian measure on $S$ specified by $\mu_0 = \mathcal N(\x_{\star}, \SSigma) \otimes \mathcal N(\0,\SSigma)$, where $\SSigma$ is a positive definite matrix in $\R^{d \times d}$. Often we take $\x_{\star} = \0$ to shorten expressions, which can be done without loss of generality by a shift in the position coordinate.
The measure $\mu_0$ will be referred to as the \emph{reference measure}. 
The Boomerang Sampler is designed in such a way that it has stationary probability distribution $\mu$ with density $\exp(-U(\x))$ relative to $\mu_0$. Equivalently, it has density 
\[ \exp \left( - U(\x) - \tfrac 1 2 (\x - \x_{\star})^{\top} \SSigma^{-1} (\x - \x_{\star}) - \tfrac 1 2 \v^{\top} \SSigma^{-1} \v \right)\]
relative to the Lebesgue measure $\dd \x \otimes \dd \v$ on $\R^d \times \R^d$. We assume that this density has a finite integral. The marginal distribution of $\mu$ with respect to $\x$ is denoted by $\Pi$.

The Boomerang process moves along deterministic trajectories $(\x_t, \v_t) \in \R^d \times \R^d$ which change direction at random times. The deterministic trajectories satisfy the following simple ordinary differential equation:
\begin{equation}
\label{eq:boomerang-ode} \frac{\dd\x_t}{\dd t} = \v_t, \quad \frac{\dd\v_t}{\dd t} = -(\x_t-\x_{\star}),
\end{equation}
with explicit solution 
$ \x_t =\x_{\star} + (\x_0 - \x_{\star}) \cos(t) + \v_0 \sin(t)$, $\v_t = -(\x_0 - \x_{\star}) \sin(t) + \v_0 \cos(t)$. Note that $(\x, \v) \mapsto \langle \x - \x_{\star}, {\bm Q}  (\x - \x_{\star}) \rangle + \langle \v, {\bm Q} \v \rangle$ is invariant with respect to the flow of~\eqref{eq:boomerang-ode} for any symmetric matrix $\bm Q$. In particular the flow of~\eqref{eq:boomerang-ode} preserves the Gaussian measure $\mu_0$ on $S$.

Given an initial position $(\x_0,\v_0) \in S$, the process moves according to the motion specified by~\eqref{eq:boomerang-ode}, resulting in a trajectory $(\x_t,\v_t)_{t \geq 0}$, until the first event occurs. The distribution of the first reflection event time $T$ is specified by
\[ \P(T \geq t) = \exp \left( - \int_0^t \lambda(\x_s,\bm v_s) \, \dd s \right),\]
where $\lambda : S \rightarrow [0,\infty)$ is the \emph{event rate} and is specified as
\begin{equation} \label{eq:reflection_intensity} \lambda(\bm x, \bm v) = \langle \v, \nabla U(\x) \rangle_+. \end{equation}
For $\x \in \R^d$ we define the \emph{contour reflection} $\bm R(\x)$ to be the linear operator from $\R^d$ to $\R^d$ given, for $(\x,\v) \in S$, by
\begin{equation} \label{eq:reflection}
\bm R(\x) \v = \v - \frac{2 \langle \nabla U(\x), \v \rangle}{|\SSigma^{1/2}\nabla U(\x)|^2} \SSigma \nabla U(\x).
\end{equation}
Importantly the reflection satisfies
\begin{equation} \label{eq:reflection_change_sign} \langle \bm R(\x) \v, \nabla U(\x) \rangle = - \langle \v, \nabla U(\x)\rangle \end{equation} and 
\begin{equation} \label{eq:reflection_norm_preservation} |\SSigma^{-1/2} \bm R(\x) \v | = |\SSigma^{-1/2} \v|, \end{equation}
which are key in establishing that the resulting Boomerang Sampler has the correct stationary distribution. 

At the random time $T$ at which a switch occurs, we put $\v_T := \bm R(\x_{T-})\v_{T-}$, where we use the notation $\y_{t-} := \lim_{s \uparrow t} \y_s$.
The process then starts afresh according to the dynamics~\eqref{eq:boomerang-ode} from the new position $(\x_T, \v_T)$.
Additionally, at random times generated by a homogeneous Poisson process with rate $\lambda_{\refr} > 0$ the velocity is refreshed, i.e. at such a random time $T$ we independently draw $\v_T \sim \mathcal N(\0, \SSigma)$. This additional input of randomness guarantees that the Boomerang Sampler can visit the full state space and is therefore ergodic, as is the case for e.g. BPS \cite{BouchardCoteVollmerDoucet2017}.

In Appendix A we define the generator of the Boomerang Sampler, which can in particular be used to prove that $\mu$ is a stationary distribution for the Boomerang process, and which can be used in subsequent studies to understand its probabilistic properties.

\begin{remark}[On the choice of the reference measure]
In principle we can express \emph{any} probability distribution $\Pi(\dd \x) \propto \exp(-E(\x))\, \dd \x$ as a density relative to $\mu_0$ by defining 
\begin{equation} \label{eq:U-in-terms-of-E} U(\x) = E(\x) - \tfrac 1 2 (\x - \x_{\star})^{\top} \SSigma^{-1} (\x - \x_{\star}).
\end{equation} As mentioned before we can take $\mu_0$ to be identical to a Gaussian prior measure in the Bayesian setting. Alternatively, and this is an approach which we will adopt in this paper, we may choose $\mu_0$ to be a Gaussian approximation of the measure $\Pi$  which may be obtained at relatively small computational cost in a preconditioning step.
\end{remark}

\subsection{Factorised Boomerang Sampler}

As a variation to the Boomerang Sampler introduced above we introduce the Factorised Boomerang Sampler (FBS), which is designed to perform well in situations where the conditional dependencies in the target distribution are sparse.
For simplicity we restrict to the case with a diagonal reference covariance $\SSigma = \mathrm{diag}(\sigma_1^2, \dots, \sigma_d^2)$.

The deterministic dynamics of the FBS are identical to those of the standard Boomerang Sampler, and given by~\eqref{eq:boomerang-ode}. The difference is that every component of the velocity has its own switching intensity. This is fully analogous with the difference between BPS and ZZ, where the latter can be seen as a factorised Bouncy Particle Sampler. In the current case, this means that as switching intensity for the $i$-th component of the velocity we take
\[ \lambda_i(\x,\v) = (v_i \partial_i U(\x))_+,\]
and once an event occurs, the velocity changes according to the operator $\bm F_i(\v)$ given by
\[ \bm F_i(\v) = \begin{pmatrix} v_1, \dots, v_{i-1}, -v_i, v_{i+1}, \dots, v_d \end{pmatrix}^{\top}.\] 
Also, the velocity of each component is refreshed according to $v_i\sim\mathcal{N}(0,\sigma_i^2)$ at rate $\lambda_{\mathrm{refr},i}>0$.

Note that the computation of the reflections has a computational cost of $\mathcal O(1)$, compared to the reflections in \eqref{eq:reflection} being at least of $\mathcal O(d)$, depending upon the sparsity of $\SSigma$. The sparse conditional dependence structure implies that the individual switching intensities $\lambda_i(\x)$ are in fact functions of a subset of the components of $\x$, contributing to a fast computation.
This feature can be exploited by an efficient `local' implementation of the FBS algorithm which reduces the number of Poisson times simulated by the algorithm (similar in spirit to the local Bouncy Particle Sampler \cite{BouchardCoteVollmerDoucet2017} and the local Zig-Zag sampler in \cite{bierkens2020piecewise}).
In Section~\ref{sec:scaling-with-dimension} we will briefly comment on the dimensional scaling of FBS. As an illustration of a realistic use, FBS will be applied to the simulation of diffusion bridges in Section~\ref{sec:diffusion-bridge}.

\subsection{Subsampling with control variates}
\label{sec:subsampling}

Let $E(\x)$ be the energy function, i.e., negative log density of $\Pi$ with respect to the Lebesgue measure.
Consider the setting where $E(\x) = \frac 1 n
\sum_{i=1}^n E^i(\x)$, as is often the case in e.g. Bayesian statistics or computational physics. (Let us stress that $n$ represents a quantity such as the number of interactions or the size of the data, and \emph{not} the dimensionality of $\x$, which is instead denoted by $d$.) Using this structure, we introduce a subsampling method using the Gaussian reference measure as a tool for the efficient construction of the Monte Carlo method.

Relative to a Gaussian reference measure with covariance $\SSigma$ centred at $\x_{\star}$, the negative log density is given by~\eqref{eq:U-in-terms-of-E}. 
Let us assume 
\begin{equation}
\label{eq:setsigma}
\SSigma = [\nabla^2 E(\x_{\star})]^{-1} \end{equation}
for a reference point $\x_{\star}$. In words, the curvature of the reference measure will agree around $\x_{\star}$ with the curvature of the target distribution.
 We can think of $\x_{\star}$ as the mean or mode of an appropriate Gaussian approximation used for the Boomerang Sampler. Note however that we shall not require that $\nabla E(\x_{\star}) = \0$ for the sampler and its subsampling alternatives to work well, although some restrictions will be imposed in Section \ref{sec:scaling-with-n}.
In this setting it is possible to employ a subsampling method which is exact, in the sense that it targets the correct stationary distribution.
This is an extension of methodology used for subsampling in other piecewise deterministic methods, see e.g. \cite{Fearnhead2016a} for an overview.

Assume for notational convenience that $\x_{\star} = \0$.
As an unbiased estimator for the log density gradient of $U$ we could simply take
\begin{equation}
\label{eq:unbiased1} 
  \widetilde{\nabla U(\x)} = 
\nabla E^I(\x) - \nabla^2 E(\0)\x,
\end{equation}
where $I$ is a random variable with uniform distribution over $\{1, \dots, n\}$. 
We shall see in Proposition~\ref{prop:lambda size} that this estimator will lead to weights which increase with $n$ and therefore we shall consider a control variate alternative.

Therefore also consider the control variate gradient estimator 
$\widehat{\nabla U(\x)}
= {\bm G}^I(\x)$, 
where, for $i = 1, \dots, n$,
\begin{equation} \label{eq:unbiased_estimator} 
\bm G^i(\x) = \nabla E^i(\x)- \nabla^2 E^i(\0) \x - \nabla E^i(\0) + \nabla E(\0).\end{equation}

Taking the expectation with respect to $I$,
\begin{align*}& \mathbb E_I \widehat {\nabla U(\x)} 
\\
&  = \frac 1 n \sum_{i=1}^n \left\{ \nabla E^i(\x)- \nabla^2 E^i(\0) \x - \nabla E^i(\0) + \nabla E(\0) \right\} \\
& = \nabla E(\x) - \nabla^2 E(\0) \x = \nabla U(\x),
\end{align*}
so that $\widehat {\nabla U(\x)}$ is indeed an unbiased estimator for $\nabla U(\x)$.
In Section \ref{sec:scaling} we shall show that $\widehat{\nabla U(\x)}$ has significantly superior scaling properties for large $n$ than $\widetilde{\nabla U(\x)}$.

\begin{remark}
In various situations we can find a reference point $\x_{\star}$ such that $\nabla E(\x_{\star}) = \0$, in which case the final term in~\eqref{eq:unbiased_estimator} vanishes. We include the term here so that it can accommodate the general situation in which $\nabla E(\x_{\star}) \neq \0$.
\end{remark}

Upon reflection, conditional on the random draw $I$, we reflect according to
\[ \bm R^I(\x) \v = \v - \frac{2 \langle \bm G^I(\x), \v \rangle}{|\SSigma^{1/2} \bm G^I(\x)|^2} \SSigma \bm G^I(\x).\]

The Boomerang Sampler that switches at the random rate $\widehat{\lambda(\x,\v)} = \langle \v, \widehat{\nabla U(\x)} \rangle_+$, and reflects according to $\bm R^I$ will preserve the desired target distribution in analogy to the argument found in \cite{BierkensFearnheadRoberts2016}.

\subsection{Simulation}
\label{sec:simulation}

The implementation of the Boomerang Sampler depends crucially on the ability to simulate from a nonhomogeneous Poisson process with a prescribed rate. In this section we will make a few general comments on how to achieve these tasks for the Boomerang Sampler and for the Subsampled Boomerang Sampler. 

Suppose we wish to generate the first event according to a switching intensity $\lambda(\x_t, \v_t)$ where $(\x_t, \v_t)$ satisfy~\eqref{eq:boomerang-ode}. This is challenging because it is non-trivial to generate points according to time inhomogeneous Poisson process, but also the function $\lambda(\x_t,\v_t)$ may be expensive to evaluate. It is customary in simulation of PDMPs to employ the technique of \emph{Poisson thinning} to generate an event according to a deterministic rate function $\overline \lambda(t) \geq 0$, referred to as \emph{computational bound}, such that $\lambda(\x_t, \v_t) \leq \overline \lambda(t)$ for all $t \geq 0$. The function $\overline \lambda(t)$ should be suitable in the sense that we can explicitly simulate $T$ according to the law
\[ \P(T \geq t) = \exp \left( - \int_0^t \overline \lambda(s) \, d s \right).\]
After generating $T$ from this distribution, we accept $T$ as a true switching event with probability $\lambda(\x_T, \v_T)/\overline \lambda(T)$. As a consequence of this procedure, the first time $T$ that gets accepted in this way is a Poisson event with associated intensity $\lambda(\x_t, \v_t)$.

In this paper we will only consider bounds $\overline \lambda(t)$ of the form $\overline \lambda(t; \x_0,\v_0) = a(\x_0,\v_0) + t b(\x_0, \v_0)$. We will call the bound \emph{constant} if $b(\x,\v) = 0$ for all $(\x,\v)$, and \emph{affine} otherwise. As a simple example, consider the situation in which $|\nabla U(\x)| \leq m$ for all $\x$. In this case we have
\[ \lambda(\x,\v) = \langle \v, \nabla U(\x) \rangle_+ \leq m |\v| \leq m \sqrt{ |\x|^2 + |\v|^2}.\]
Since the final expression is invariant under the dynamics~\eqref{eq:boomerang-ode}, we find that
\[ \lambda(\x_t,\v_t) \leq m \sqrt{ |\x_0|^2 + |\v_0|^2}, \quad t \geq 0,\]
which gives us a simple constant bound.

In the case of subsampling the switching intensity $\widehat {\lambda(\x,\v)}$ is random. Still, the bound $\overline \lambda(t; \x_0, \v_0)$ needs to be an upper bound for all random realizations of $\widehat {\lambda(\x,\v)}$. In the case we use the unbiased gradient estimator $\widehat{ \nabla U(\x)} = \bm G^I$ of~\eqref{eq:unbiased_estimator}, we can bound e.g.
\[ \widehat {\lambda(\x,\v)} \leq \sup_{i,\x} |\bm G^i(\x)| |\v| \leq \sup_{i,\x} |\bm G^i(\x)| \sqrt{|\x|^2 +|\v|^2},\]
assuming all gradient estimators $\bm G^i$ are globally bounded. We will introduce different bounds in detail in Appendix B.

\section{Scaling for large data sets and large dimension}
\label{sec:scaling}

\subsection{Robustness to large $n$}
\label{sec:scaling-with-n}

In this section, we shall investigate the variability of the rates induced by the Boomerang Sampler and its subsampling options. The size of these rates is related to the size of the upper bounding rate Poisson process used to simulate them. Moreover, the rate of the upper bounding Poisson rate is proportional to the number of density  evaluations, which in turn is a sensible surrogate for the computing cost of running the algorithm.

As in Section \ref{sec:subsampling}, we describe $E$ as a sum of $n$
constituent negative log-likelihood terms: $E(\x)=-\sum_{i=1}^n \ell_i (\x )$. 
(In the notation above we are just setting $\ell_i (\x )=-nE_i(\x)$.)
Under suitable regularity conditions, the target probability measure $\Pi$ satisfies posterior contraction around $\x=\0$ at the rate $\eta $, that is
for all $\epsilon$ there exists $\delta >0$ such that $\Pi (B_{n^{-\eta } \delta }(\0)) >1- \epsilon$ where $B_r(\0)$ denotes the ball of radius $r$ centred at $\0$. 
As a result of this, we typically have velocities of order
$n^{-\eta}$
ensuring that the dynamics in (\ref{eq:boomerang-ode})
circles the state space in $\mathcal O(1)$ time.

The various algorithms will have computational times roughly proportional to the number of likelihood evaluations, which in turn depends on the event rate (and its upper bound). Therefore we shall introduce explicitly the subsampling bounce rates corresponding to the use of the unbiased estimators in (\ref{eq:unbiased1}) and (\ref{eq:unbiased_estimator}).
\[
\widetilde {\lambda(\bm x, \bm v)} = \langle \v, \widetilde{\nabla U(\x)} \rangle_+\ ;
\quad \quad 
\widehat {\lambda(\bm x, \bm v)} = \langle \v, \widehat{\nabla U(\x)} \rangle_+ \ .
\]
To simplify the arguments below, we also assume that $\ell _i$ has all its third derivatives uniformly bounded, implying that all third derivative terms of $E$ are bounded by a constant multiple of $n$. This allows us to write down the expansion 
\begin{align}
 \nabla U (\x) &= \nabla E (\0) + \nabla^2 E(\0) \x -\SSigma ^{-1}\x + \mathcal O(n|\x |^2) \nonumber \\   
 &=\nabla E (\x) \nonumber \\
 &= \nabla E (\0)   + \mathcal O(n|\x |^2)\ .
\end{align}
Similarly we can write
\begin{align}
 \widehat{\nabla U (\x)} &= n \nabla \ell _I (\x ) -n\nabla^2 \ell_I ({\bf 0}) \x - n \ell_I({\bf 0}) \nonumber \\
 &+ \nabla E (\0) -\SSigma^{-1}\x
 \nonumber \\ 
 &=
\nabla E (\0)   + \mathcal O(n|\x |^2)\ .
 \end{align}
 using the same Taylor series expansion.

We can now use this estimate directly to obtain bounds on the event rates. We summarise this discussion in the following result.

\begin{proposition}
\label{prop:lambda size}
Suppose that $\x, \ \v \in B_{n^{-\eta } \delta }(\0)$ for some $\delta $, and under the assumptions described above, we have that
\begin{align}
\lambda (\x,\v ) &\le 
{\cal O}\left(
n^{-\eta}
(|\nabla E(\0)|+n^{1-2\eta })
\right)
\\    
\widetilde{ \lambda(\x,\v)} &\le 
{\cal O}\left(
|\nabla E (\0)|) + {\cal O}(n)
\right)
\\  
\widehat{\lambda(\x,\v)} &\le 
{\cal O}\left(
n^{-\eta}
(|\nabla E(\0)|+n^{1-2\eta })
\right)
\end{align}
\end{proposition}
Thus the use of ${\widehat {\nabla U(\x )}} $ does not result in an increased event rate (in order of magnitude). There is therefore an ${\cal O}(n)$ computational advantage obtained from using subsampling due to each target density valuation being ${\cal O}(n)$ quicker.

Proposition \ref{prop:lambda size} shows that as long as the reference point $\x^*$ (chosen to be ${\bf 0}$ here for convenience) is chosen to be sufficiently close to the mode so that $|\nabla E (\0)|$ is at most $\mathcal O(n^{1-2\eta})$, then we have that
$$
{\lambda(\x,\v)}= \widehat{\lambda(\x,\v)} =
{\cal O}\left(
n^{1-3\eta }\right)\ .
$$
Note that this rate can go to $0$ when $\eta >1/3$.  In particular in the regular case where Bernstein von-Mises theorem holds, we have $\eta = 1/2$. In this case the rate of jumps for the Boomerang can recede to $0$ at rate $n^{-1/2}$ so long as $|\nabla E (\0)|$ is at most $\mathcal O(1)$).

\subsection{Scaling with dimension}
\label{sec:scaling-with-dimension}

In this section, we will discuss how the Boomerang Sampler has an attractive scaling property for high dimension. This property is qualitatively similar to the preconditioned Crank-Nicolson algorithm \cite{MR1723510, MR2444507} and the elliptical slice sampler \cite{pmlr-v9-murray10a} which take advantage of the reference Gaussian distribution.

The dimensional complexity of BPS and ZZ was studied in \cite{bierkens2018highdimensional, Deligiannidis2018,Andrieu2018}. For the case of an isotropic target distribution,
the rate of reflections per unit of time is constant for BPS and proportional to $d$ for ZZ with unit speeds in all directions. On the other hand, the time until convergence is of order $d$ for the BPS and $1$ for ZZ. Therefore, the total number of reflections required for convergence of these two algorithms is of the same order which grows linearly with dimension.

For the Boomerang Sampler we consider the following  setting. Consider reference measures $\mu_{0,d}(\dd \x, \dd \v) = \mathcal N(\bm 0, \SSigma_d) \otimes \mathcal N(\bm 0, \SSigma_d)$ for increasing dimension $d$, where for every $d = 1, 2, \dots$, $\SSigma_d$ is a $d$-dimensional positive definite matrix.
Relative to these reference measures we consider a sequence of potential functions $U_d(\x)$.
Thus relative to Lebesgue measure our target distribution $\Pi_d(\dd \x)$ has density $\exp(-E_d(\x))$, where $E_d(\x) = U_d(\x) + \tfrac 1 2 \langle  \x,\SSigma_d^{-1} \x \rangle$. Let $\E_d$ denote expectation with respect to $\Pi_d(\dd \x) \otimes \mathcal N(0, \SSigma_d)(\dd \v)$.
We assume that the sequence $(U_d)$ satisfies
\begin{equation}
\label{eq:potentials-condition}
\sup_{d=1,2,\ldots}\E_d[|\SSigma^{1/2}_d\nabla U_d(\x)|^2]<\infty,
\end{equation}
The condition~\eqref{eq:potentials-condition} arises naturally for instance in the context of Gaussian regression, spatial statistics,  Bayesian inverse problems as well as the setting of the diffusion bridge simulation example described in detail in Section \ref{sec:diffusion-bridge}.

Furthermore we assume that the following form of the Poincar\'e inequality holds,
\begin{equation} \label{eq:poincare-inequality}
\E_d[f_d(\x)^2]\leq
\frac 1 {C^2} \E_d[|\SSigma_d^{1/2}\nabla f_d(\x)|^2]
\end{equation}
with constant $C > 0$ independent of dimension, and where $f_d : \R^d \rightarrow \R$ is any mean zero differentiable function.
A sufficient condition for~\eqref{eq:poincare-inequality} to hold is 
\[
C^2I\preceq \SSigma_d^{1/2}\nabla^2E_d(\x)\SSigma_d^{1/2}=\SSigma_d^{1/2}\nabla^2 U_d(\x)\SSigma_d^{1/2}+I
\]
by the classical Brascamp-Lieb inequality \cite{MR0450480,Bakry_2014};  note that it may also hold in the non-convex case, see e.g. \cite{LorenziBertoldi2007}, Section 8.6.

Under the stated assumptions we argue that (i) the expected number of reflections per unit time scales as $\mathcal O(1)$ with respect to dimension, and (ii) within a continuous time interval that scales as $\mathcal O(1)$, the Boomerang Sampler mixes well. 
Claims (i) and (ii) are provided with a heuristic motivation in the Appendix. A rigorous proof for this claim remains part of our future work.

In the ideal but non-sparse scenario, the computational cost of the event time calculation for the Boomerang Sampler is thus expected to be a factor $d$ smaller compared to BPS and ZZ assuming that the cost per event is the same for these algorithms.
However, this comparison is unrealistic since in general we can not simulate reflections directly. In practice, we need to use the thinning method as discussed in Section \ref{sec:simulation}. The thinning method introduces a significant amount of shadow events (which are rejected after inspection), and the true events usually represent a small portion relative to the number of shadow events. As a result there can be a high cost for calculating shadow events even when the number of true events is small. 

For the FBS, the expected number of events per unit of time is $\sum_{i=1}^d\mathbb{E}_d[(v_i\partial_i U(\x))_+]$. Under the hypothesis above, this is of  $\mathcal O(d^{1/2})$. Thus, the number of events is much bigger than that of the Boomerang. However, as in the case of ZZ, under a sparse model assumption, the cost of calculation per jump is of constant order whereas it is of the order of $d$ for the Boomerang Sampler. Therefore, the Factorised Boomerang Sampler should outperform the Boomerang Sampler for this sparse setup.

\section{Applications and experiments}

\subsection{Logistic regression}
\label{sec:logistic}

As a suitable test bed we consider the logistic regression inference problem. 
Given predictors $\y^{(1)}, \dots, \y^{(n)}$ in $\R^d$, and outcomes $z^{(1)}, \dots, z^{(n)}$ in $\{0,1\}$, we define the log likelihood function as
\begin{align*} & \ell(\x) = -\sum_{i=1}^n \left\{ \log(1 + e^{\x^{\top} \y^{(i)}})- z^{(i)} \x^{\top} \y^{(i)}  \right\}.\end{align*}
Furthermore we impose a Gaussian prior distribution over $\x$ which for simplicity we keep fixed to be a standard normal distribution throughout these experiments. As a result we arrive at the negative log target density
\[ E(\x) = \sum_{i=1}^n \left\{ \log(1 + e^{\x^{\top} \y^{(i)}})- z^{(i)} \x^{\top} \y^{(i)}  \right\} + \tfrac 1 2 \x^{\top} \x.\]
As a preprocessing step when applying the Boomerang Sampler, and all subsampled methods, we find the mode $\x_{\star}$ of the posterior distribution and define $\SSigma$ by~\eqref{eq:setsigma}. We apply the Boomerang Sampler, with and without subsampling.
These samplers are equipped with an affine computational bound and a constant computational bound respectively, both discussed in Appendix B (the affine bound is usually preferred over a constant bound, but a useful affine bound is not available in the subsampling case).

We compare the Boomerang to both BPS and ZZ with and without subsampling. In all subsampling applications we employ appropriate control variance techniques to reduce the variability of the random switching intensities, as discussed in Section~\ref{sec:subsampling}. Furthermore in the dimension dependent study we include the Metropolis adjusted Langevin algorithm (MALA) for comparison. Throughout these experiments we use Effective Sample Size (ESS) per second of CPU time as measure of the efficiency of the methods used. ESS is estimated using the Batch Means method, where we take a fixed number of 50 batches for all our estimates. ESS is averaged over the dimensions of the simulation and then divided by the runtime of the algorithm to obtain ``average ESS per second'' (other ESS summaries could also have been used). The time horizon is throughout fixed at $10,000$ (with 10,000 iterations for MALA). 
For ZZ and BPS the magnitude of the velocities is rescaled to be comparable on average with Boomerang, to avoid unbalanced runtimes of the different algorithms. 
In Figures~\ref{fig:logistic-by-observations} and~\ref{fig:logistic-by-dimensions} the boxplots are taken over 20 randomly generated experiments, where each experiment corresponds to a logistic regression problem with a random (standard normal) parameter, based on randomly generated data from the model.\footnote{The code used to carry out all of the experiments of this paper may be found online at \url{https://github.com/jbierkens/ICML-boomerang}.}
 The refreshment rates for BPS and the Boomerang Samplers are taken to be 0.1.

\begin{figure}[ht!]
\includegraphics[width = \columnwidth]{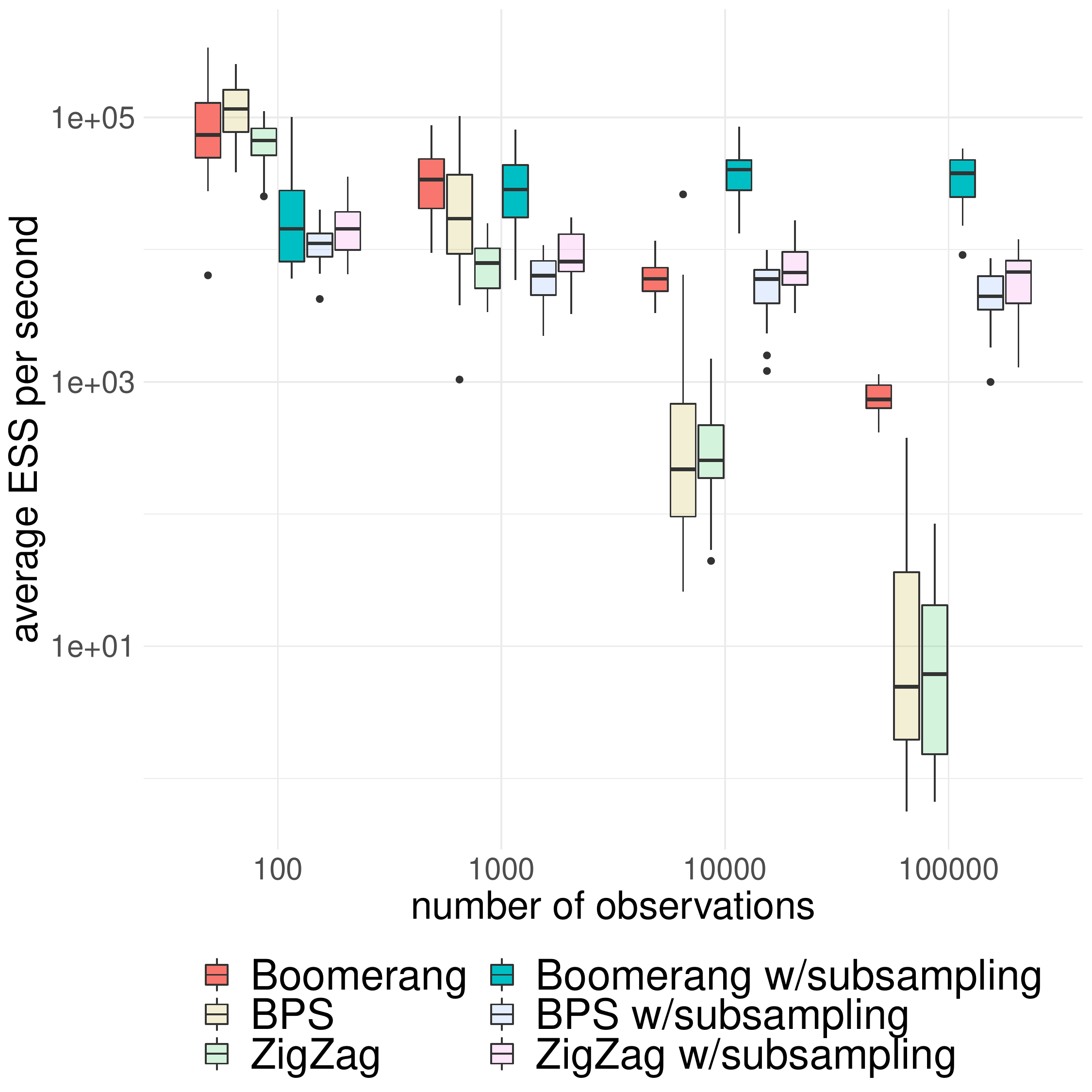}
\caption{Scaling of Boomerang Sampler compared to other PDMC methods for the logistic regression problem of Section~\ref{sec:logistic} as a function of the number of observations. Here $d = 2$.}
\label{fig:logistic-by-observations}
\end{figure}

The Boomerang Sampler is seen to outperform the other algorithms, both in terms of scaling with dimension as with respect to an increase in the number of observations. For a fixed dimension, the subsampling algorithms will clearly outperform the non-subsampling algorithms as number of observations $n$ grows. In particular, the ESS/sec stays fixed for the subsampled algorithms, and decreases as $\mathcal O(n)$ for the non-subsampled versions. In this case, we did not include the MALA algorithm since we observed its complexity strongly deteriorating as the number of observations increases. For a large number of observations ($n \geq 10,000$, $d = 2$) we see that the Boomerang Sampler (with and without subsampling) accepts almost none of the proposed switches. This means that effectively we are sampling from the Gaussian reference measure. This observed behaviour is in line with the scaling analysis in Section~\ref{sec:scaling-with-n}.

In the second experiment we let the dimension $d$ grow for a fixed number of observations. The subsampling algorithms currently do not scale as well as the non-subsampled versions. For practical purposes we therefore only consider non-subsampled algorithms for the comparison with respect to dimensional dependence. For the dimensions $d \leq 32$ we tested the Boomerang outperforms MALA, but it seems empirically that MALA has a better scaling behaviour with dimension. Note that MALA needs careful tuning to exhibit this good scaling.
We remark that the beneficial scaling properties of the underlying Boomerang Process as discussed in Section~\ref{sec:scaling-with-dimension} may be adversely affected by suboptimal computational bounds. We are optimistic that the dimensional scaling of subsampled algorithms can be further improved by designing better computational bounds. 

In all cases the necessary preprocessing steps can be done very quickly. In particular the plots are not affected by including (or excluding) the preprocessing time in the computation of ESS/sec.

\begin{figure}[ht!]
\includegraphics[width = \columnwidth]{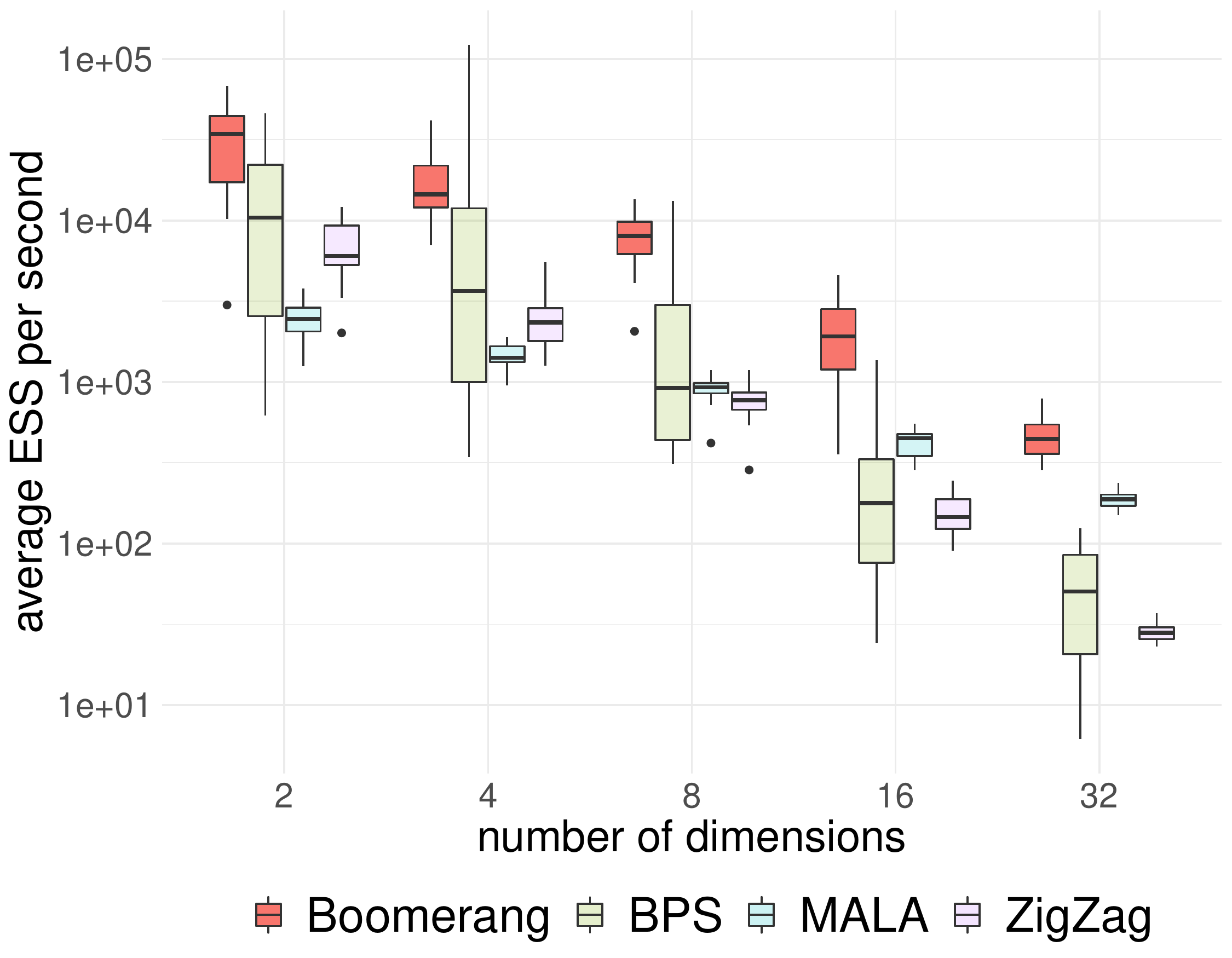}
\caption{Scaling of Boomerang Sampler compared to other PDMC methods and MALA for the logistic regression problem of Section~\ref{sec:logistic} as a function of the number of dimensions. Here the number of observations is $n = 1,000$.}
\label{fig:logistic-by-dimensions}
\end{figure}

\subsection{Diffusion bridges}
\label{sec:diffusion-bridge}
In \cite{bierkens2020piecewise} the authors introduce a framework for the simulation of diffusion bridges (diffusion processes conditioned to hit a prescribed endpoint) taking strong advantage of the use of factorised piecewise deterministic samplers. 
This invites the use of the Factorised Boomerang Sampler (FBS). We consider time-homogeneous one-dimensional conditional diffusion processes (diffusion bridges) of the form 
\[
\dd X_t = b(X_t) \dd t + \dd W_t, \quad X_0 = u, \quad X_T = v
\]
where $W$ is a scalar Brownian motion and $b$ satisfies some mild regularity conditions (see \cite{bierkens2020piecewise} for details). This simulation problem is an essential building block in Bayesian analysis of non-linear diffusion models with low frequency observations \cite{10.2307/2673434}. 

We consider the approach of \cite{bierkens2020piecewise} where the diffusion path on $[0,T]$ is expanded with a truncated Faber Schauder basis as
\[
X^N_t = \Bar{\Bar\phi}(t) u + \Bar{\phi}(t) v + \sum_{i = 0}^N \sum_{j = 0}^{2^i - 1} \phi_{i,j}(t) x_{i,j}.
\]
Here, 
\begin{align*} \Bar \phi(t) & = t/T, \quad \quad  \Bar{\Bar\phi}(t)  = 1 - t/T, \\ \phi_{0,0}(t) & =\sqrt{T}\big ( (t/T)\ind_{[0,T/2]}(t) +(1-t/T)\ind_{(T/2,T]}(t)\big), \\
\phi_{i,j}(t) & = 2^{-i/2}\phi_{0,0}(2^i t - jT) \quad i \geq 0, \quad  0\le j\le 2^i-1, \end{align*}
 are the Faber-Schauder functions and $N$ is the truncation of the expansion. In \cite{bierkens2020piecewise}, ZZ is used to sample the corresponding coefficients $\x := (x_{0,0},...,x_{N,2^N-1})$ which have a density measure written with respect to a standard Gaussian reference measure (corresponding to a Brownian bridge measure in the path space, see \cite{bierkens2020piecewise} for details). In particular we have that
\begin{equation}
    \label{target bridge measure}
    \frac{\dd \mu}{\dd \mu_0}(\x, \v) \propto \exp \left\{ -\frac{1}{2}\int_0^T \left( b^2(X^N_s) + b'(X^N_s)\right)  \dd s \right\}
\end{equation}
where $b'$ is the derivative of $b$ and $\mu_0 = \mathcal{N}(\bm{0}, \bm I ) \otimes \mathcal{N}(\bm{0}, \bm I )$ with $\bm I$ the $2^{N+1}-1$ dimensional identity matrix. The measure given by  \eqref{target bridge measure} has a remarkable conditional independence property (Proposition 2, \cite{bierkens2020piecewise}) and the coefficients $x_{i,j}$, for $i$ large, responsible for the local behaviour of the process, are approximately independent standard Gaussian, reflecting the fact that, locally, the process behaves as a Brownian motion.

In \cite{bierkens2020piecewise} the authors device a local implementation of ZZ which optimally exploits the sparse conditional independence structure of the target distribution, alleviating the computational costs in high dimensional setting (e.g. of a high truncation level  $N$).
Since the Girsanov density~\eqref{target bridge measure} is expressed relative to a standard normal distribution on the coefficients $\x$, the Factorised Boomerang Sampler is a natural candidate for a further reduction in computational cost, by reducing the required number of simulated events, in particular at the higher levels where the coefficients have approximately a Gaussian distribution. This will allow for a further increase of the truncation level $N$ and/or faster computations at a fixed truncation levels.

We consider the the class of diffusion bridges with drift equal to 
\begin{equation}
    \label{drift sin diffusion}
    b(x) = \alpha \sin (x), \, \alpha \ge 0. 
\end{equation}
 The higher $\alpha$, the stronger is the attraction of the diffusion paths to the stable points $(2k-1)\pi, k \in \mathbb{N}$ while for $\alpha = 0$ the process reduces to a Brownian bridge with $\mu = \mu_0$. Equivalently to  \cite{bierkens2020piecewise}, we use subsampling as the gradient of the log density in \eqref{target bridge measure} involves a time integral that cannot be solved analytically in most of the cases. The unbiased estimator for $\partial_{x_{i,j}} U(\x)$ is the integrand evaluated at a uniform random point multiplied by the range of the integral. The Poisson bounding rates used for the subsampling can be found in the Appendix E.

Figure \ref{sin_10} shows the resulting bridges for $\alpha = 1$, starting at $u = -\pi$ and hitting $v = 3\pi$ at final time $T = 50$ after running the FBS with clock $T^\star = 20000$, as simulated on a standard desktop computer. The refreshment rate relative to each coefficient $x_{i,j}$ is fixed to $\lambda_{\mathrm{refr},i,j} = 0.01$ and the truncation of the expansion is $N = 6$.

\begin{figure}[h]
    \centering
    \includegraphics[width = 8cm]{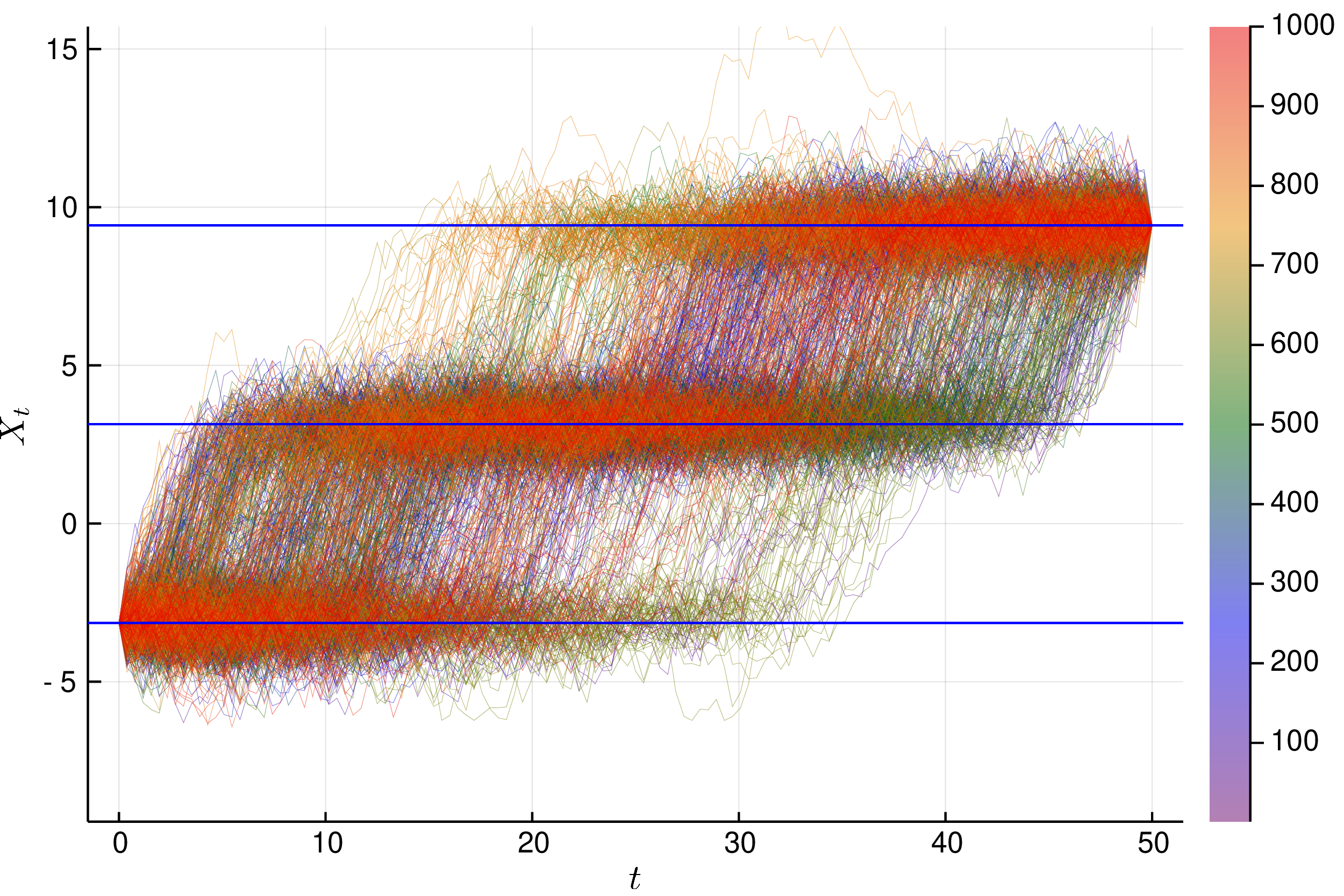}
    \caption{1000 diffusion bridges with drift equal to \eqref{drift sin diffusion} with $\alpha = 1, \, u = -\pi, \, v = 3\pi, \, T = 50, L = 6$ sampled with the FBS with time horizon $T^\star = 20,000$ and  refreshment rates $\lambda_{\mathrm{refr},i}= 0.01$ for all $i$. The straight horizontal lines are the attraction points.}
    \label{sin_10}
\end{figure}

In Figure \ref{comparisons}, we compare the performances of the Boomerang Sampler and ZZ by computing the average number of reflections ($y$-axis on a log-scale) for the coefficients $x_{i,j}$ at each level ($x$-axis). The number of reflections is understood as a measure of complexity of the algorithm. We repeat the experiment for $\alpha = 0.5$ and $\alpha = 0 $ (where $\mu = \mu_0$) and fix the truncation level to be $N=10$ which corresponds to a $2047 + 2047$ dimensional space for ($\x,\v$). For a fair comparison we set the expected $\ell^1$ norm of the velocities and the time horizon of the two samplers to be the same. In both cases, the average number of reflections converges to the average number of reflections under $\mu_0$ (dashed lines) indicating that the coefficients at high levels are approximately standard normally distributed but while ZZ requires a fixed number of reflections for sampling from $\mu=\mu_0$, the Boomerang does not, allowing to high resolutions of the diffusion bridges at lower computational cost.

\begin{figure}[h]
    \centering
    \includegraphics[width = 8cm ]{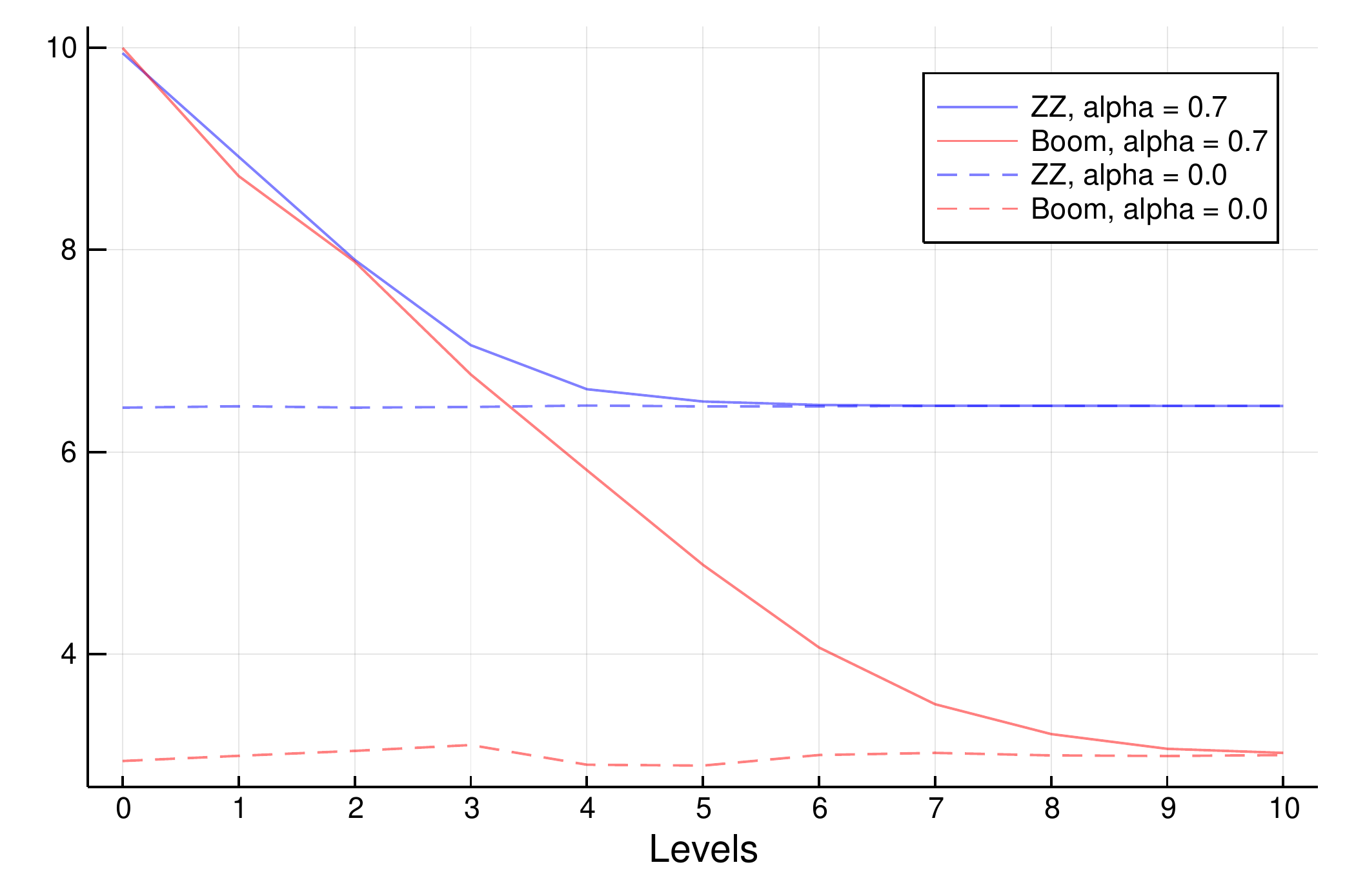}
    \caption{Average number of reflections (on a log-scale) for the coefficients $x_{i,j}$ at the level $i =0,1,..,10$ for the diffusion bridge given by \eqref{drift sin diffusion} with $\alpha = 0.5$ (solid lines) and $\alpha = 0.0$ (dashed lines) for the Zig-Zag Sampler (blue lines)  and the Factorised Boomerang Sampler (red lines) with $T^\star = 2,000$ and Boomerang refreshment rates $\lambda_{\mathrm{refr},i} = 0.01$ for all $i$.}
    \label{comparisons}
\end{figure}

\subsection{Dependence upon reference measure}

In a final experiment we investigate the dependence of the performance of the Boomerang Sampler upon the choice of reference measure.
For this we consider a simple setting in which the target distribution is a standard normal distribution in $d$ dimensions.
However, instead of using the standard normal distribution as reference measure, we perturb it in two ways: (i) we vary the component-wise variance $\sigma^2$ of the reference measure, and (ii) we vary the mean $\x_{\star}$ of the reference measure.
Specifically, we choose a reference measure $\mathcal N(\x_{\star}, \SSigma) \otimes \mathcal N(\bm 0, \SSigma)$, which we choose in case (i) to be $\x_{\star} = \bm 0, \SSigma = \sigma^2 \bm I$, and in case (ii), $\x_{\star} = \alpha (1, \dots, 1)^{\top}$, $\SSigma = \bm I$.
As performance measure we use the ESS per second for the quantity $|\x|^2$. We use refreshment rate $0.1$ for Boomerang, and we compare to the Bouncy Particle Sampler, with refreshment $1.0$, with both samplers run over a time horizon of 10,000.
In Figure~\ref{fig:perturbed-reference-sigma} the results of this experiment are displayed for varying $\sigma^2$, and in Figure~\ref{fig:perturbed-reference-xref} the results are displayed for varying $\x_{\star}$. The box plots are taken over 20 experiments of the Boomerang Sampler, which are compared to a single run of  the Bouncy Particle Sampler (dashed line).

\begin{figure}[ht!]
\centering
\includegraphics[width = 8 cm]{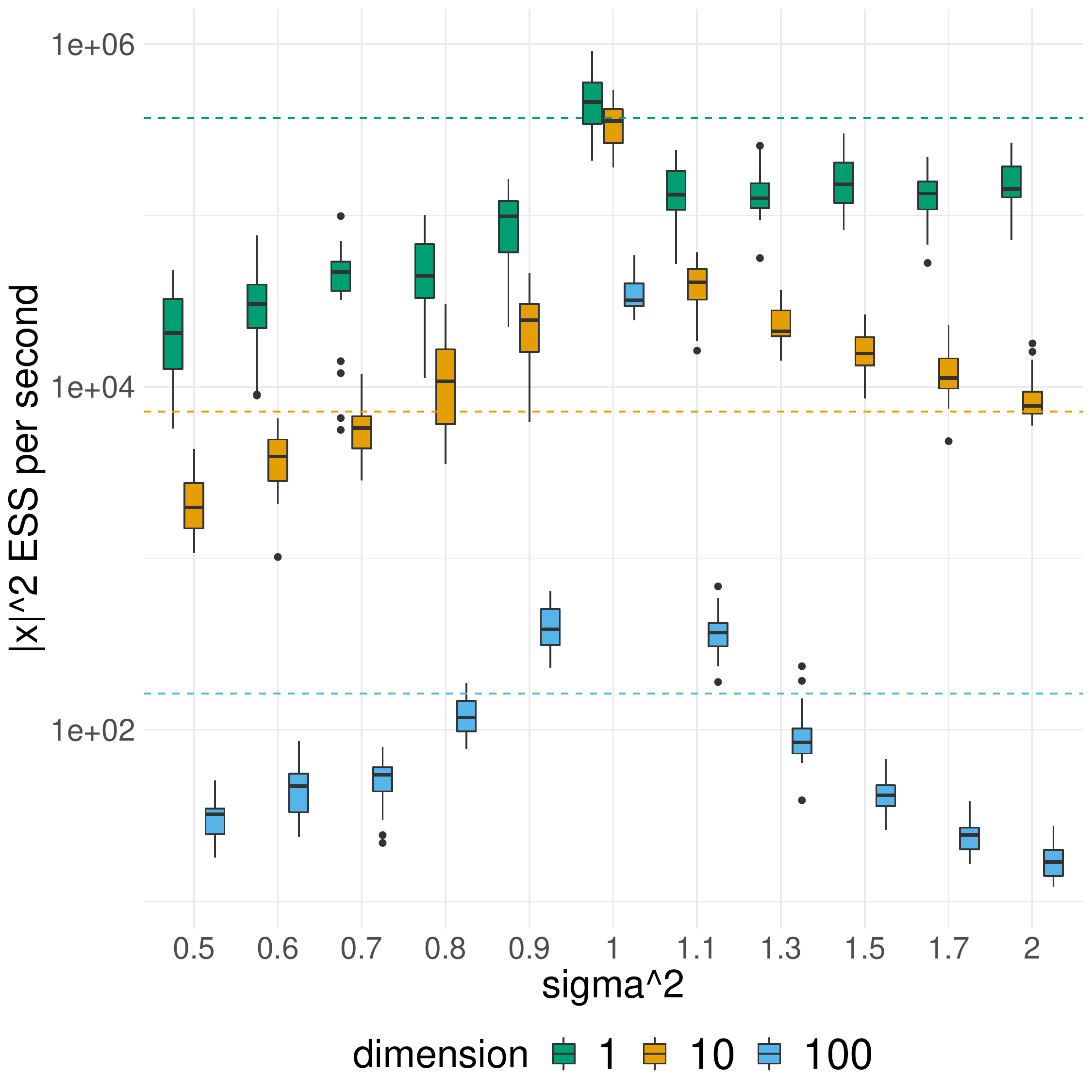}
\caption{Effect of perturbing the variance of the reference measure. As reference measure we choose $\mathcal N(\0, \sigma^2 \bm I) \otimes \mathcal N(\0, \sigma^2 \bm I)$, where $\sigma^2$ is varied from $0.5$ to $2.0$. }
\label{fig:perturbed-reference-sigma}
\end{figure}

\begin{figure}[ht!]
\centering
\includegraphics[width = 8 cm]{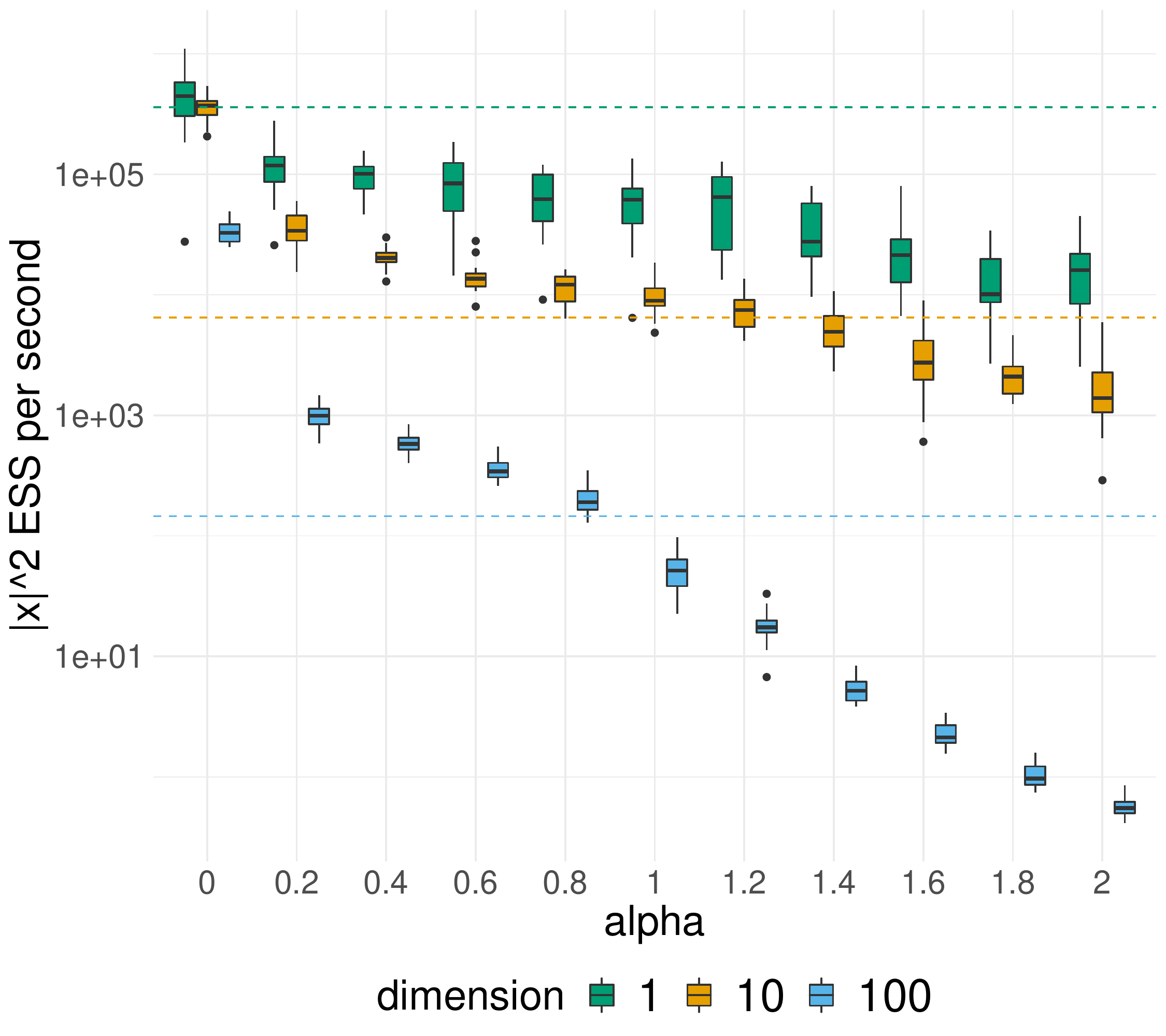}
\caption{Effect of perturbing the mean of the reference measure. As reference measure we choose $\mathcal N(\alpha \bm 1, \bm I) \otimes \mathcal N(\0, \bm I)$, where $\alpha$ is varied from $0.0$ to $2.0$. }
\label{fig:perturbed-reference-xref}
\end{figure}

In this setting, the Boomerang Sampler significantly outperforms the BPS, although the performance is seen to depend upon the choice of reference measure.
Note however that the dependencies of $\SSigma$ on $\sigma^2$ and of $\x_{\star}$ upon $\alpha$ scale as $\operatorname{trace} \SSigma = \sigma^2 d$ and $\|\x_{\star}\| = \alpha d^{1/2}$ respectively, so that in high dimensional cases the sensitivity on $\x_{\star}$ and $\SSigma$ may be more moderate than might appear from Figures~\ref{fig:perturbed-reference-sigma} and~\ref{fig:perturbed-reference-xref}.

\section{Conclusion}

We presented the Boomerang Sampler as a new and promising methodology, outperforming other piecewise deterministic methods in the large $n$, moderate $d$ setting, as explained theoretically and by performing a suitable benchmark test. The theoretical properties of the underlying Boomerang Sampler in high dimension are very good. However currently a large computational bound and therefore a large number of rejected switches are hampering the efficiency. We gave a numerical comparison which demonstrates that Boomerang performs well against its natural competitors; however one should be cautious about drawing too many conclusions about the performance of the Boomerang without a more comprehensive simulation study. Further research is required to understand in more detail the dependence upon e.g. reference covariance $\SSigma$, centering position $\x_{\star}$, refreshment rate, computational bounds and the choice of efficiency measure. 

We furthermore introduced the Factorised Boomerang Sampler and illustrated its ability to tackle a challenging simulation problem using an underlying sparse structure.

An important direction for further research is the improvement of the computational bounds, in particular with the aim of having good scaling with dimension of subsampled algorithms. Related to this it is important to gain a better understanding of the relative optimality of subsampled versus non-subsampled algorithms in the large $n$, large $d$ case. 

\section*{Acknowledgements}
We thank the anonymous reviewers and metareviewer for their helpful comments, which have helped to improve the paper.
JB and SG acknowledge funding under the research programme `Zig-zagging through computational barriers' with project number 016.Vidi.189.043, which is financed by the Netherlands Organisation for Scientific Research (NWO). KK is supported by JSPS KAKENHI Grant Number 20H04149 and JST CREST Grant Number JPMJCR14D7.  GR is supported by EPSRC grants
EP/R034710/1 (CoSInES) and EP/R018561/1 (Bayes for Health), and by Wolfson merit award WM140096.

\appendix

\section{Generator and stationary distribution}
\subsection{Boomerang Sampler}

For simplicity take $\x_{\star} = \0$. The generator of the Boomerang Sampler is defined by
\begin{align*}
\mathcal{L}\psi(\x,\v)&=\langle \v,  \nabla_{\x}\psi(\x,\v)\rangle -\langle \x,  \nabla_{\v}\psi(\x,\v)\rangle\\ &\quad+\lambda(\x,\v)\left(\psi(\x,\bm R(\x)\v)-\psi(\x,\v)\right)\\
&\quad+\lambda_{\mathrm{refr}}\left(\int_{\R^d}\psi(\x,\w)\phi(\w)\, \dif \w-\psi(\x,\v)\right), 
\end{align*}
for any compactly supported differentiable function $\psi$ on $S$, where $\phi$ is the probability density function of $\mathcal{N}(\bm 0,\SSigma)$. 

Taking $\lambda(\x,\v)$ and $\bm R(\x)$ as in Eqs. (2) and (3) of the paper respectively, we will now verify that $\int_S \mathcal L \psi \, d \mu = 0$ for all such functions $\psi$, and for $\mu$ being the measure on $S$ with density $\exp(-U(\x))$ relative to $\mu_0$. This then establishes that the Boomerang Sampler has stationary distribution $\mu$. A complete proof also requires verification that the compactly supported, differentiable functions form a core for the generator, which is beyond the scope of this paper. For a discussion of this topic for archetypal PDMPs see \cite{Holderrieth2019}.

First we consider the terms involving the partial derivatives of $\psi$. By partial integration, we find
\begin{align*}
    & \int_S \langle \v, \nabla_{\x}\psi(\x,\v)\rangle -\langle \x, \nabla_{\v}\psi(\x,\v)\rangle \, \mu(\dd \x, \dd \v) \\
    & = \int_S \psi(\x,\v) \langle \v, \nabla U(\x)\rangle \, \mu(\dd \x, \dd \v)
\end{align*}

Next we inspect the term representing the switches occurring at rate $\lambda(\x,\v)$. By Eq. (5) of the paper, the coordinate transform $\w = \bm R(\x) \v$ (for fixed $\x$) leaves the measure $\mathcal N(\0,\SSigma)$ over the velocity component invariant. Using this observation, we find that \begin{align*}
& \int_S \lambda(\x,\v) (\psi(\x, \bm R(\x) \v) - \psi(\x, \v)) \,\mu (\dd \x, \dd \v) \\
& = \int_S \lambda(\x, \bm R(\x) \w) \psi(\x, \w) \,\mu(d \x, d \w) \\
& \quad - \int_S \lambda(\x, \v) \psi(\x, \v) \, \mu(\dd \x, \dd \v) \\
& = \int_S [ \lambda(\x, \bm R(\x) \v) - \lambda(\x, \v)] \psi(\x, \v) \, \mu(\dd \x, \dd \v).
\end{align*}
Using Eq.~(2) and (4) of the paper, and the identity $(-a)_+ - (a)_+ = -a$, it follows that this expression is equal to
\begin{align*}
    & \int_S \left[ \langle \bm R(\x) \v, \nabla U(\x)\rangle_+ - \langle \v, \nabla U(\x) \rangle_+ \right] \psi(\x, \v) \, \mu(\dd \x, \dd \v) \\
    & = -\int_S \langle \v, \nabla U(\x)\rangle   \psi(\x, \v) \, \mu(\dd \x, \dd \v).
\end{align*}

Finally by changing the order of integration, it can be shown that
\[ \int_S \lambda_{\mathrm{refr}}\left(\int_{\R^d}\psi(\x,\v)\phi(\v)\, \dif \v-\psi(\x,\v)\right) \,  \mu_0(\dd \x, \dd \v) = 0.\]

Adding all terms yields that $\int_S \mathcal L \psi \, \dd \mu = 0$.

\subsection{Factorised Boomerang Sampler}
The Factorised Boomerang Sampler has  generator
\begin{align*}
    \mathcal L\psi(\x, \v) & = \langle \v, \nabla_{\x} \psi(\x,\v) \rangle - \langle \x, \nabla_{\v} \psi(\x, \v) \rangle \\
    & \quad + \sum_{i=1}^d \lambda_i(\x, \v) (\psi(\x, \bm F_i(\v)) - \psi(\x, \v)) \\
    & \quad + \lambda_{\refr} \left( \int \psi(\x, \w)  \phi(\w) \, d \w  -\psi(\x, \v) \right).
\end{align*}
Verifying stationarity of $\mu$ is done analogously to the case of the non-factorised Boomerang Sampler, but now has to be carried out componentwise.

\section{Computational bounds}

Suppose $(\x_t, \v_t)$ satisfies the Hamiltonian dynamics ODE of Eq. (1) in the paper, starting from $(\x_0, \v_0)$ in $\R^d \times \R^d$.
Throughout we assume $U : \R^d \rightarrow \R$ is a twice continuously differentiable function with Hessian matrix $\nabla^2 U$.
Furthermore we assume without loss of generality that $\x_{\star} = \bm 0$. 
First we consider bounds for switching intensities of the form $\lambda(\x,\v) = \langle \v, \nabla U(\x) \rangle_+$. For a matrix $\bm A \in \R^{d \times d}$ we use $\|\bm A\|$ to denote the matrix norm induced by the Euclidean metric.

\begin{lemma}[Constant bound]
\label{lem:constant-bound}
Suppose there exists a constant $M > 0$ such that for all $\x \in \R^d$ we have the global bound
\[ \| \nabla^2 U(\x)\| \leq M. \]
 Define $m := |\nabla U(\0)|$.
Then for all $t \geq 0$,
\begin{equation}
    \label{eq:affine-bound} \lambda(\x_t,\v_t) \leq \frac{M}{2} (|\x_0|^2 + |\v_0|^2)  + m \sqrt{|\x_0|^2 + |\v_0|^2}.
\end{equation}
\end{lemma}

\begin{proof}
We have the following estimate on the switching intensity.
\begin{align*}
    & \lambda(\x,\v)  = \langle \v, \nabla U(\x) \rangle_+ \\
    & \leq \langle \v, \nabla U(\bm 0)\rangle_+ + \int_0^1 |\langle \v, \nabla^2 U(\x s) \bm x\rangle| \,  \dd s.
\end{align*}
We may bound the inner product in the integrand as follows.
\begin{align*}
|\langle \v, \nabla^2 U(\bm y) \bm x \rangle|
&\le \|\nabla^2 U(\bm x)\|~|\v|~|\bm x|\\
&\le M~\left(\frac{|\v|^2+|\bm x|^2}{2}\right)
\end{align*}
by the Cauchy–Schwarz inequality. 
Also \begin{align*}
& |\langle \v, \nabla U(\bm 0) \rangle| \leq m |\v| \leq m  \sqrt{|\x|^2 + |\v|^2}.
\end{align*}
Combining these estimates and the fact that $|\x_t|^2 + |\v_t|^2$ is invariant under the dynamics of Eq. (1) in the paper yields the stated result.
\end{proof}

\begin{lemma}[Affine bound]
\label{lem:affine-bound}
Suppose $\|\nabla^2 U(\x)\| \leq M$ for all $\x \in \R^d$, and let $m = |\nabla U(\0)|$. Then for a solution $(\x_t, \v_t)$ to Eq. (1) of the paper with $\lambda(\x,\v) = \langle \v, \nabla U(\x) \rangle_+$, we have for all $t \geq 0$
\begin{align*}
    \lambda(\x_t, \v_t)\leq \left(a(\x_0, \v_0) +  t b(\x_0, \v_0) \right)_+,
\end{align*}
where 
\begin{align*}
 a(\x,\v) & = \langle \v, \nabla U(\x)\rangle_+, \quad \text{and} \\
 b(\x,\v) & = M \left(|\x|^2+|\v|^2\right) + m \sqrt{|\x|^2+|\v|^2}.
 \end{align*}
\end{lemma}

\begin{proof}
By the Hamiltonian dynamics,
\begin{align*}
   &  \frac{\dd}{\dd t} \langle \v_t, \nabla U(\x_t) \rangle \\
   & = - \langle \x_t, \nabla U(\x_t) \rangle + \langle \v_t, \nabla^2 U(\x_t) \v_t \rangle \\
    & = -\langle \x_t, \nabla U(\bm 0) \rangle - \int_0^1 \langle \x_t, \nabla^2 U(s \x_t) \x_t \rangle \, \dd s \\
    & \quad \quad + \langle \v_t, \nabla^2 U(\x_t) \v_t \rangle \\
    & \leq |\x_t| |\nabla U(\bm 0) | + M \left( |\x_t|^2 + | \v_t|^2 \right).
\end{align*}
Using that $|\x_t|^2 + |\v_t|^2$ is invariant under the dynamics yields the stated result.
\end{proof}

\begin{lemma}
\label{lem:bounded_gradient}
Suppose $|\nabla U(\y)| \leq C$ for all $\y \in \R^d$. Then, for all trajectories $(\x_t, \v_t)$ satisfying Eq.~(1) of the paper we have
\[ \lambda(\x_t, \v_t) \leq C \sqrt{|\x_0|^2 + |\v_0|^2}.\]
\end{lemma}
\begin{proof}
We have
\[ \lambda(\x, \v) \leq C |\v| \leq C \sqrt{|\x|^2 + |\v|^2},\]
and the latter expression is constant along trajectories.
\end{proof}

Analogously we have the following useful bound for the Factorized Boomerang Sampler.
\begin{lemma}
\label{lem:factorized_bound_constant_gradient}
Suppose $U : \R^d \rightarrow \R$ is differentiable. Suppose there exist constants $c_1, \dots, c_d$ such that, for all $\y \in \R^d$ and $i = 1, \dots, d$, we have 
\[
|\partial_i U(\x)| \leq c_i \quad \text{for all $\x$, $i$}.
\]
Then 
\[ \lambda_i(\x_t, \v_t) \le c_i \sqrt{|x^i_0|^2 + |v^i_0|^2}.\]
\end{lemma}

\begin{lemma}
\label{lem:factorized_affine_bound}
Suppose for all $i$ we have that
\[
    \sqrt{\sum_j \partial_i \partial_j U(\x)^2} \le  M_i, 
    \]
and
\[ |\partial_i U(\0)| \leq m_i.\]
    Then 
    \[\lambda_i(\x_t, \v_t) \le  \left(a_i(\x_0, \v_0) + b_i(\x_0, \v_0)  t\right)^+\]
    where 
    \[
    a_i(\x, \v) = (v^{i} \partial_i U(\x))^+
    \]
    \begin{align*}
    & b_i(\x, \v)\\
    & =\sqrt{(x^i)^2 + (v^i)^2} \left ( m_i + M_i \sqrt{|\x|^2 + |\v|^2} \right).\end{align*}
\end{lemma}

 \begin{proof}
We compute
\begin{align*}
 & \frac{\dd}{\dd t}v^i_t \partial_i U(\x_t) \\ 
& = -x^i_t \partial_i U(\x_t) + v^i_t \sum_{j=1}^d \partial_i \partial_j U(\x_t)v^j_t \\
& = -x_t^i \partial_i U(\0) - \int_0^1 x^i_t \sum_{j=1}^d\partial_i \partial_j U(s\x_t )x^j_t \dd s \\ & \quad \quad + v^i_t \sum_{j=1}^d  \partial_i \partial_j U(\x_t)v^j_t \\
& \leq \sqrt{(x_t^i)^2 + (v_t^i)^2} |\partial_i U (\0)| + M_i |x_t^i| |\x_t| + M_i |v_t^i| |\v_t|\\
& \leq \sqrt{(x_t^i)^2 + (v_t^i)^2} |\partial_i U (\0)| \\
&  \quad + M_i/2 \left(  \alpha( |x_t^i|^2 + |v_t^i|^2)  +  (1/\alpha)(|\x_t|^2 +  |\v_t|^2) \right).
\end{align*}
Optimising over $\alpha$, and using that $|x^i_t|^2 + |v^i_t|^2$ is constant along Factorised Boomerang Trajectories, yields the stated result.
\end{proof}    

\subsection{Computational bounds for subsampling}
\label{sec:subsampling_computationalbound}
In the case of subsampling we use the unbiased estimator of Eq.~(9) of the paper. 

\begin{lemma}
\label{lem:computational_bound_subsampling}
Suppose that for some positive definite matrix $\bm Q$ we have 
that, for all $i$, and $\y_1, \y_2 \in \R^d$, 
\begin{equation} \label{eq:condition-Q-subsampling}
    \nabla^2 E^i(\y_1) - \nabla^2 E^i(\y_2) \preceq {\bm Q},
\end{equation}
where $A\preceq B$ means $B-A$ is positive semidefinite.
Suppose $\widehat{\nabla  U(\x)}$ is given by Eq.~(9) of the paper, and $\nabla E(\0) = \0$. Along a trajectory $(\x_t, \v_t)$ satisfying the Hamiltonian dynamics of Eq.~(1) of the paper, we have, for all $t \geq 0$, that
\[ \langle \v_t, \widehat{\nabla U(\x_t)} \rangle \leq \tfrac 1 2 (|{\bm Q}^{1/2} \x_0|^2 +|{\bm Q}^{1/2} \v_0|^2 ), \quad \text{a.s.}\]
where the almost sure statement is with respect to all random (subsampling) realisations of the switching intensity.
\end{lemma}

\begin{remark}
Lemma~\eqref{lem:computational_bound_subsampling} is easily extended to the case in which $\nabla E(\0) \neq \0$. In this case we have
\begin{align*} \langle \v_t, \widehat{U(\x_t)} \rangle & \leq \tfrac 1 2 (|{\bm Q}^{1/2} \x_0|^2 +|{\bm Q}^{1/2} \v_0|^2 ) \\
& \quad \quad \quad + (|\v_0|^2 + |\x_0|^2)^{1/2} |\nabla E(\0)|, \quad \text{a.s.}\end{align*}
\end{remark}

\begin{remark}
In practice one may wish to take $\bm Q$ to be a diagonal matrix, which reduces the computation of the computational bound to a $\mathcal O(d)$ computation instead of $\mathcal O(d^2)$. For example one could take $\bm Q = c \bm I$ for a suitable constant $c > 0$ such that~\eqref{eq:condition-Q-subsampling} is satisfied.
\end{remark}

\begin{remark}[Affine bound  for subsampling is strictly worse]
When we try to obtain an affine bound, of the form 
\[ \widehat{\lambda(\x_t, \v_t)} \leq a(\x_0, \v_0) + b(\x_0,\v_0),\]
then it seems we cannot avoid an expression for $a$ of the form of the bound in Lemma~\ref{lem:computational_bound_subsampling}. As a consequence, the affine bound is strictly worse than the constant bound.
\end{remark}

\begin{proof}[Proof (of Lemma~\ref{lem:computational_bound_subsampling})]
Suppose we have $I = i$ for the random index $I$ in Eq.~(9) of the paper. We compute
\begin{align*}
  &   \langle \v_t, \widehat{\nabla U(\x_t)} \rangle \\
  & = \langle \v_t, \nabla E^i(\x_t) - \nabla^2 E^i(\0) \x_t- \nabla E^i(\0) \rangle \\
  & = \langle \v_t,\int_0^1 \nabla^2 E^i(s \x_t) \x_t \, \dd s - \nabla^2 E^i(\0) \x_t \rangle.
\end{align*}
Then we may continue the above computation to find, using Lemma~\ref{lem:matrix-inequality} below, that 
\begin{align*}
    \langle \v_t, \widehat{\nabla U(\x_t)} \rangle & = \int_0^1 \langle \v_t,[\nabla^2 E^i(s \x_t) - \nabla^2 E^i(\0)] \x_t \rangle \, \dd s  \\
    & \leq \int_0^1 |{\bm Q}^{1/2} \v_t| \, |{\bm Q}^{1/2} \x_t| \, \dd s \\
    & \leq \tfrac 1 2 (|{\bm Q}^{1/2} \v_t|^2 + |{\bm Q}^{1/2} \x_t|^2 ).
\end{align*}

Since $\tfrac 1 2 (|{\bm Q}^{1/2} \v_t|^2 + |{\bm Q}^{1/2} \x_t|^2)$ is invariant under the dynamics, the stated conclusion follows.
\end{proof}

\begin{lemma}
\label{lem:matrix-inequality}
Suppose $\bm M, \bm P \in \R^{d \times d}$ are symmetric matrices with $\bm P$ positive definite and such that $-\bm P \preceq \bm M \preceq \bm P$. Then $\langle \bm M \y, \z \rangle \leq  |\bm P^{1/2} \y| \, |\bm P^{1/2} \z|$ for all $\y, \z \in \R^{d \times d}$.
\end{lemma}

\begin{proof}
Taking $\bm y = \bm P^{-1/2} \x$, we find
\begin{align*}
|\langle \bm P^{-1/2} \bm M \bm P^{-1/2} \x, \x \rangle| = |\langle \bm M \y, \y \rangle| \leq \langle \bm P \y, \y \rangle = |\x|^2,
\end{align*}
which establishes that $\| \bm P^{-1/2} \bm M \bm P^{-1/2} \| \leq 1$.
Using this observation  we arrive at
\begin{align*}
    \langle \bm M \y, \z \rangle & \leq \underbrace{\|\bm P^{-1/2} \bm M \bm P^{-1/2} \|}_{\leq 1}\,  |\bm P^{1/2} \y| \, |\bm P^{1/2} \z|.
\end{align*}
\end{proof}

\section{Scaling with dimension}
In Section~3.2 of the paper, we discuss the scaling of the Boomerang Sampler with dimension. The argument in that section is self contained, but relies on the observation that the change of $E_d(\x_t)$ over a time interval of order 1 is at least of order $d^{1/2}$. Here we motivate this observation.

In the following arguments, we assume stationarity of the process for simplicity.
Let $U_d$, $\SSigma_d$, $E_d$, $\Pi_d$, $\E_d$ be as described in Section 3.2 of the manuscript. For simplicity and without loss of generality we assume that $E_d(\x)$ is normalised as $\E_d[E_d(\x)]=0$. Furthermore, for simplicity we assume that $\E_d[\x]=\0$ although this condition can be relaxed.
 
As discussed we suppose 
that the sequence $(U_d)$ satisfies
\begin{equation}
\label{eq:condition-gradient-U}
\sup_{d\in\mathbb{N}}\E_d[|\SSigma_d^{1/2}\nabla U_d(\x)|^2]\le \kappa
\end{equation}
for some $\kappa>0$. 
Furthermore, we assume that the following form of the Poincar\'e inequality is satisfied for $\Pi_d(\dd\x) \propto \exp(-E_d(\x))\dif\x$:
\begin{equation} \label{eq:poincare}
C~\E_d\left[f_d(\x)^2\right]^{1/2}\le 
\E_d\left[|\SSigma_d^{1/2}\nabla f_d(\x)|^2\right]^{1/2}
\end{equation}
 for some constant $C>0$ not depending on $d$, and any differentiable  function $f_d: \R^d \rightarrow \R$ with mean $0$ and finite variance. 

By~\eqref{eq:condition-gradient-U} the expected number of reflections per unit time
$\E_d[\langle\v,\nabla U_d(\x)\rangle_+]$ is bounded with respect to dimension.
However the process mixes well in a single time unit under suitable regularity conditions
as we will discuss now.

By applying~\eqref{eq:poincare} to $f_d(\x)=(\SSigma_d^{-1/2}\x)_i$, where $\v_i$ denotes the $i$-th coordinate of $\v$, we have 
$C^2\E_d[|\SSigma_d^{-1/2} \x|^2]\le\E_d[\operatorname{trace}(\SSigma_d^{-1/2}\SSigma_d\SSigma_d^{-1/2})]=d$,
using the stated assumption $\E_d[\x] = \0$.

Also by~\eqref{eq:poincare} and by Minkowski's inequality, 
\begin{align*}
\E_d[E_d(\x)^2]^{1/2}&
\le C^{-1}
\E_d[|\SSigma_d^{1/2}\nabla E_d(\x)|^2]^{1/2}\\
&= C^{-1}
\E_d[|\SSigma_d^{1/2}\nabla U_d(\x)+\SSigma_d^{-1/2}\x|^2]^{1/2}\\
&= C^{-1}(\kappa^{1/2}+C^{-1}d^{1/2}) = \mathcal O(d^{1/2}).
\end{align*}
If $(\x_t, \v_t)$ satisfies the ODE Eq.~(1) of the paper, the unit time difference $E_d(\x_t)-E_d(\x_0)$ is 
\begin{align*}
\int_0^t\langle \nabla E_d(\x_s),\v_s\rangle\dif s\approx \int_0^t\langle  \SSigma_d^{-1}\x_s,\v_s\rangle\dif s.
\end{align*}
Here, the difference between the left- and the right-hand sides is $\int_0^t\langle\SSigma_d^{1/2}\nabla U(\x_s),\SSigma_d^{-1/2}\v_s\rangle\dif s$ which is  of order $d^{1/2}$ under the assumption of stationarity by ~\eqref{eq:condition-gradient-U} and the Cauchy-Schwarz inequality, using that $\E_d[|\SSigma_d^{-1/2}\v_s|^2]=d$.
The right-hand may be simplified to
\begin{align*}
&\int_0^t\langle  \SSigma_d^{-1}(\x_0\cos s+\v_0\sin s),-\x_0\sin s+\v_0\cos s\rangle\dif s\\
&=A_0\int_0^t2\sin s~\cos s~\dif s+B_0\int_0^t(\cos^2s-\sin^2s)\dif s\\
&=A_0(1-\cos 2t)/2+B_0(\sin 2t)/2
\end{align*}
where $A_0=(\langle \v_0,\SSigma_d^{-1}\v_0\rangle -\langle \x_0,\SSigma_d^{-1} \x_0 \rangle)/2$
and $B_0=\langle \x_0,\SSigma_d^{-1}\v_0\rangle$. 
Then $A_0$ and $B_0$ are 
uncorrelated since $\SSigma_d^{-1/2}\v_0$ follows the standard normal distribution. 
Also, $\E_d[A_0^2]\ge \operatorname{Var}(A_0)\ge \operatorname{Var}(\langle \v_0,\SSigma_d^{-1}\v_0\rangle)=2d$. Therefore, 
\begin{align*}
\E_d[|E_d(\x_t)-E_d(\x_0)|^2]&\gtrsim 
\E_d[A_0^2]\left(\frac{1-\cos 2t}{2}\right)^2\\
&\ge 2d\left(\frac{1-\cos 2t}{2}\right)^2. 
\end{align*}

Thus the change of $E_d(\x_t)$ over a term interval of $\mathcal O(1)$ is of order $d^{1/2}$ whereas $E_d(\x_t)$ itself has the same order. These informal arguments suggest that dynamics of the Boomerang sampler in a finite time interval sufficiently changes the log density even in high dimension. However, further study should be made in this direction. 

\section{Logistic regression}

We assume a prior distribution $\pi_0(\x) \sim \mathcal N(\0,\sigma^2 \bm I)$ on $\R^d$. Given predictors $\y^{(1)}, \dots, \y^{(n)}$ in $\R^d$, and outcomes $z^{(1)}, \dots, z^{(n)}$ in $\{0,1\}$, we obtain the negative log posterior distribution as
\begin{align*} & E(\x) = \sum_{i=1}^n \left\{ \log(1 + e^{\x^{\top} \y^{(i)}})- z^{(i)} \x^{\top} \y^{(i)}  \right\} +  |\x|^2/2 \sigma^2.\end{align*}
We then have
\begin{align*}
 \nabla E(\x) & = \x / \sigma^2 + \sum_{i=1}^n \y^{(i)} \left[ \frac {e^{\x^{\top} \y^{(i)}}} { 1+ e^{\x^{\top} \y^{(i)}}} -z^{(i)} \right], \\
 \nabla^2 E(\x) & =\bm I /\sigma^2 + \sum_{i=1}^n \frac{\y^{(i)} (\y^{(i)})^{\top}e^{\x^{\top} \y^{(i)}}}{  \left( 1+ e^{\x^{\top} \y^{(i)}}\right)^2}.
\end{align*}
In the experiments in this paper we take a flat prior, i.e. $\sigma^2 = \infty$.

Let 
\[ \x_{\star} = \argmin_{\x \in \R^d} E(\x).\]
We take $\SSigma^{-1} = \nabla^2 E(\x_{\star})$. We have $U(\x) = E(\x) - (\x - \x_{\star})^{\top} \nabla^2 E(\x_{\star}) (\x -\x_{\star})/2$, which is a difference of two positive definite matrices. Using the general inequality $a \mapsto |a|/(1+a)^2 \leq 1/4$, we find
\[ - \tfrac 1 4 \sum_{i=1}^n \y^{(i)} (\y^{(i)})^{\top} \preceq \nabla^2 U(\x) \preceq \tfrac 1 4 \sum_{i=1}^n \y^{(i)} (\y^{(i)})^{\top}.\]
We then simply have
\[ \| \nabla^2 U(\y) \| \leq M :=\tfrac 1 4 \| \sum_{i=1}^n \y^{(i)} (\y^{(i)})^{\top}\|.\] 

These observations may be applied in conjunction with the lemmas of Appendix B to obtain useful constant and affine computational bounds for the switching intensities.

\section{Diffusion bridge simulation}
We consider diffusion bridges of the form
\begin{equation} \label{eq:sde-sin}
\dd X_t = \alpha \sin(X_t) \dd t + \dd W_t, \quad X_0 = u, X_T = v, t\in [0,T]
\end{equation}
where $W$ is a scalar Brownian motion and $\alpha \ge 0$. The diffusion path is expanded with a truncated Faber-Schauder basis such that
\[
X^N_t = \Bar{\Bar\phi}(t) u + \Bar{\phi}(t) v + \sum_{i = 0}^N \sum_{j = 0}^{2^i - 1} \phi_{i,j}(t) x_{i,j},
\]
where $N$ is the truncation of the expansion and 
\begin{align*} \Bar \phi(t) & = t/T, \quad \quad  \Bar{\Bar\phi}(t)  = 1 - t/T, \\ \phi_{0,0}(t) & =\sqrt{T}\big ( (t/T)\ind_{[0,T/2]}(t) +(1-t/T)\ind_{(T/2,T]}(t)\big), \\
\phi_{i,j}(t) & = 2^{-i/2}\phi_{0,0}(2^i t - jT) \quad i \geq 0, \quad  0\le j\le 2^i-1,
\end{align*}
are the Faber-Schauder functions. As shown in \cite{bierkens2020piecewise}, the measure of the coefficients corresponding to~\eqref{eq:sde-sin} is derived from the Girsanov formula and given by
\[
\frac{\dd \mu}{\dd \mu_0}(\x, \v) \propto \exp \left\{ \frac{-\alpha}{2}\int_0^T \left( \alpha \sin^2(X^N_s) + \cos(X^N_s)\right)  \dd s \right\}
\]
where $\mu_0 = \mathcal{N}(\bm{0}, \bm I ) \otimes \mathcal{N}(\bm{0}, \bm I )$ with $\bm I$ the $2^{N+1}-1$ dimensional identity matrix. 
By standard trigonometric identities we have that
\[
\partial_{x_{i,j}}U(\x) = \frac{\alpha}{2} \int_{S_{i,j}} \phi_{i,j}(t) \left( \alpha \sin\left(2X^{N}_t\right) - \sin\left(X^{N}_t \right)\right) \dd t
\]
where $S_{i,j}$ is the support of the basis function $\phi_{i,j}$. Similarly to \cite{bierkens2020piecewise}, for each $i,j$, we use subsampling and consider the unbiased estimator for $\partial_{x_{i,j}}U(\x)$ given by   
\[
\widehat{\partial_{x_{i,j}}U(\x)} = S_{i,j}  \phi_{i,j}(\tau_{i,j}) \left( \alpha^2 \sin\left(2X^{N}_{\tau_{i,j}}\right) - \alpha \sin\left(X^{N}_{\tau_{i,j}} \right)\right)
\]
where $\tau_{i,j}$ is a uniform random variable on $S_{i,j}$. This gives Poisson rates $\widehat{\lambda_{i,j}(\x,\v)} = \langle \v, \widehat{\partial_{x_{i,j}} U(\x)} \rangle_+$. 
In this case, for all $i,j$,  $|\widehat{\partial_{x_{i,j}} U(\x)}|$ is globally bounded, say by $m_{i,j}$. We use the constant Poisson bounding rates given, in similar spirit as in Section 2.3 of the paper, by
\[
\overline \lambda_{i,j}(\x_t,\v_t) = m_{i,j}\sqrt{|x^{i,j}_0|^2 + |v^{i,j}_0|^2}, 
\]
where we used that $t \rightarrow |x^{i,j}_t|^2 + |v^{i,j}_t|^2$ is constant under the Factorised Boomerang trajectories.  
Similarly to \cite{bierkens2020piecewise}, the FBS gains computational efficiency by a local implementation which exploits the fact that each $\overline \lambda_{i,j}(\x,\v)$ is a function of just the coefficient $x_{i,j}$ (see \cite{bierkens2020piecewise}, Algorithm 3, for an algorithmic description of the local implementation of a factorised PDMP).


\end{document}